\documentclass[acmtog]{acmart}

\usepackage[usestackEOL]{stackengine}
\usepackage{scalerel}
\usepackage{graphicx,amsmath}
\def\myupbracefill#1{\rotatebox{90}{\stretchto{\{}{#1}}}
\def\rlwd{.5pt}
\newcommand\notate[4][B]{%
  \if B#1\else\def\myupbracefill##1{}\fi%
  \def\useanchorwidth{T}%
  \setbox0=\hbox{$\displaystyle#2$}%
  \def\stackalignment{c}\stackunder[-6pt]{%
    \def\stackalignment{c}\stackunder[-1.5pt]{%
      \stackunder[2pt]{\strut $\displaystyle#2$}{\myupbracefill{\wd0}}}{%
    \rule{\rlwd}{#3\baselineskip}}}{%
  \strut\kern9pt$\rightarrow$\smash{\rlap{$~\displaystyle#4$}}}%
}

\newcommand{\bA}{\mathbf{A}}
\newcommand{\bB}{\mathbf{B}}
\newcommand{\bC}{\mathbf{C}}

\newcommand{\bF}{\mathbf{F}}

\newcommand{\be}{\mathbf{e}}

\newcommand{\bg}{\mathbf{g}}
\newcommand{\bz}{\mathbf{z}}

\newcommand{\bn}{\mathbf{n}}
\newcommand{\bp}{\mathbf{p}}
\newcommand{\bs}{\mathbf{s}}

\newcommand{\bu}{\mathbf{u}}
\newcommand{\bv}{\mathbf{v}}
\newcommand{\bx}{\mathbf{x}}

\newcommand{\bq}{\mathbf{q}}

\newcommand{\bX}{\mathbf{X}}

\newcommand{\bzero}{\mathbf{0}}

\newcommand{\dgds}{\frac{\partial \bg}{\partial \bs}}

\newcommand{\phases}{\mathbf{v}}

\newcommand{\levelset}{\mathbf{\phi}}

\newcommand{\deformed}{\mathbf{x}}
\newcommand{\enriched}{\hat{\bx}}

\newcommand{\eigvec}{\mathbf{v}}
\newcommand{\eigval}{\lambda}

\usepackage{booktabs} 
\usepackage[thinc]{esdiff}
\usepackage{wrapfig}
\usepackage{xcolor}
\usepackage[ruled]{algorithm2e} 

\usepackage{amsthm}

\newtheorem{theorem}{Theorem}[section]

\SetAlFnt{\small}
\SetAlCapFnt{\small}
\SetAlCapNameFnt{\small}
\SetAlCapHSkip{0pt}
\IncMargin{-\parindent}

\setcopyright{acmlicensed}
\acmJournal{TOG}
\acmYear{2023} \acmVolume{42} \acmNumber{4} \acmArticle{1} \acmMonth{8} \acmPrice{15.00}\acmDOI{10.1145/3592114}


\setcitestyle{square}

\newcommand{\Changes}[1]{\textcolor{black}{#1}}

\begin{document}



\title{Differentiable Stripe Patterns 
\\for Inverse Design of Structured Surfaces}

\author{Juan Montes Maestre}
\affiliation{%
  \institution{ETH Z{\"u}rich}
  \country{Switzerland}
}
\email{jmontes@inf.ethz.ch}

\author{Yinwei Du}
\affiliation{%
  \institution{ETH Z{\"u}rich}
  \country{Switzerland}
}
\email{yinwei.du@inf.ethz.ch}

\author{Ronan Hinchet}
\affiliation{%
  \institution{ETH Z{\"u}rich}
  \country{Switzerland}
}
\email{ronanhinchet@gmail.com}

\author{Stelian Coros}
\affiliation{%
  \institution{ETH Z{\"u}rich}
  \country{Switzerland}
}
\email{stelian@inf.ethz.ch}

\author{Bernhard Thomaszewski}
\affiliation{%
  \institution{ETH Z{\"u}rich}
  \country{Switzerland}
}
\email{bthomasz@ethz.ch}

\renewcommand{\shortauthors}{Montes et al.}

\begin{abstract}

Stripe patterns are ubiquitous in nature and everyday life. While the synthesis of these patterns has been thoroughly studied in the literature, their potential to control the mechanics of structured materials remains largely unexplored. In this work, we introduce Differentiable Stripe Patterns---a computational approach for automated design of physical surfaces structured with stripe-shaped bi-material distributions. 
Our method builds on the work by Kn{\"o}ppel and colleagues \shortcite{Knoeppel15Stripe} for generating globally-continuous and equally-spaced stripe patterns. To unlock the full potential of this design space, we propose a gradient-based optimization tool to automatically compute stripe patterns that best approximate macromechanical performance goals. Specifically, we propose a computational model that combines solid shell finite elements with XFEM for accurate and fully-differentiable modeling of elastic bi-material surfaces. To resolve non-uniqueness problems in the original method, we furthermore propose a robust formulation that yields unique and differentiable stripe patterns. 
We combine these components with equilibrium state derivatives into an end-to-end differentiable pipeline that enables inverse design of mechanical stripe patterns. 
We demonstrate our method on a diverse set of examples that illustrate the potential of stripe patterns as a design space for structured materials. Our simulation results are experimentally validated on physical prototypes.
\end{abstract}

%
%

%
%



\begin{teaserfigure}
\includegraphics[width=7.0in]{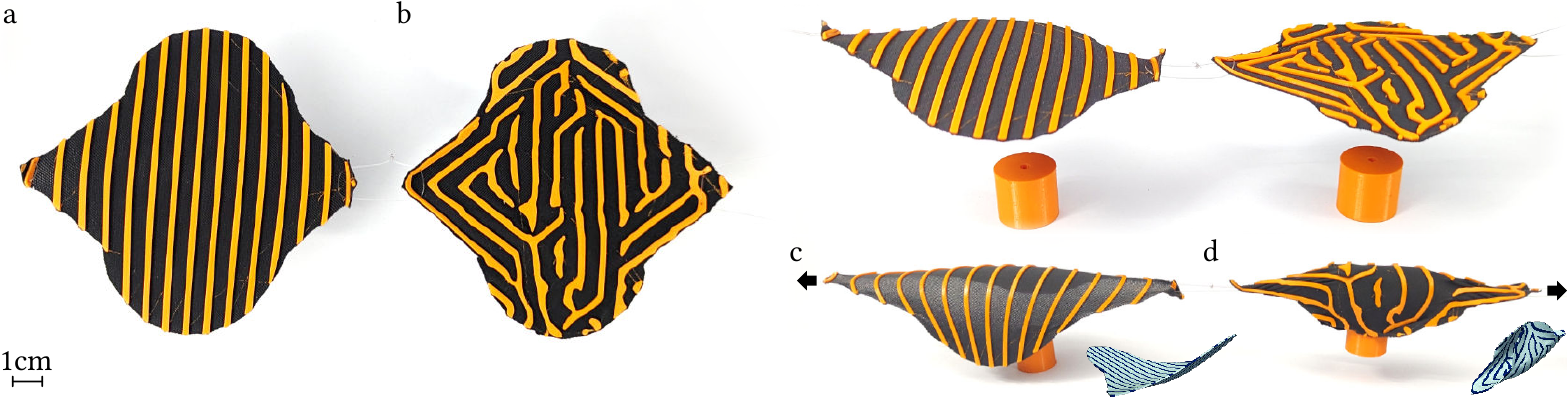}
 \caption{Optimization of a compliant gripper using Differentiable Stripe Patterns. An initial design with parallel stripes (\textit{a}) and an optimized design (\textit{b}) deforms out of plane under actuation, but the initial design does not produce sufficient closure (\textit{c}). Our optimized design (\textit{d}) produces significantly larger lateral deflection for the same actuation in simulation, as shown in the rendered insets. The corresponding physical prototype confirms this prediction and is able to enclose and lift a small object. }
 \label{fig:teaser}
\end{teaserfigure}

\maketitle
\section{Introduction}

From soft to stiff, from isotropic to anisotropic, and from homogeneous to functionally-graded---designing materials with tailored mechanical properties is a central problem in many fields of science and engineering. Structured materials are particularly interesting in this context since their macromechanical behavior can be controlled through their microscale geometry.
Here we consider a particular class of structured surfaces that are quasi-inextensible in a given direction while being compliant in the orthogonal direction. By varying the principal directions across the surface, these materials can achieve a broad range of macromechanical effects, making them interesting for applications in, e.g., sportswear, orthotics, and robotics (see Fig. \ref{fig:teaser}). Designing structured surfaces that lead to desired mechanical behavior, however, is a challenging problem.

In this work, we propose Differentiable Stripe Patterns---a computational approach to performance-oriented design of structured surfaces. We draw inspiration from the work of Kn{\"o}ppel et al. \shortcite{Knoeppel15Stripe} who proposed a method for generating equally-spaced, globally-continuous stripe patterns on arbitrary surfaces. Our key observation is that, when interpreted as bi-material distributions, stripe patterns form an ideal design space for structured surfaces with high stiffness contrast.
To unlock the full potential of this material space, we envision an inverse design tool that automatically computes 
stripe patterns that lead to an optimal approximation of high-level performance goals.

To implement this vision, we must overcome a number of challenges. First, stripe patterns are generated from tangent vector fields by solving a complex eigenvalue problem (EVP). Due to intrinsic symmetries, eigenvalues of this EVP always occur in pairs. This nontrivial multiplicity means that eigenvectors are not unique---they form an eigenplane---and derivatives do not exist.
Second, predicting the mechanics of bi-material distributions with macroscopic stripe patterns requires accurate modeling of material interfaces. While conforming discretizations are an obvious choice for static interfaces, finite changes in stripe patterns require remeshing to maintain valid tesselations. Such discrete changes are highly problematic for gradient-based optimization. 
Third, to obtain valid stripe patterns, eigenvectors must be further processed with vector normalizations and other nonsmooth operations. Derivatives of these operations are discontinuous or diverge at singularities, which is again highly problematic for optimization-based design.

Our Differentiable Stripe Patterns integrate dedicated solutions to each of these problems:
\begin{itemize}
\item We achieve robust evolution of eigenvectors by decomposing the problem of finding the \textit{best} vector from an eigenplane into the much simpler sub-problems of finding \textit{a}) \textit{some} reference vector from the eigenplane and \textit{b}) the \textit{optimal} in-plane rotation with respect to this vector. This unique and differentiable parameterization of the eigenplane yields well-defined stripe pattern derivatives.
\item To allow for continuous motion of material interfaces without remeshing, we propose a computational model that combines solid shells with the extended finite element method (XFEM).
Solid shells are made of prismatic elements that derive their response to bending from differential stretching through the thickness. Unlike discrete bending models common in graphics, solid shells follow standard finite element theory and thus integrate seamlessly with XFEM.
\item We remove and replace non-smooth operations and introduce a design-space regularizer that prevents numerical singularities with logarithmic barrier functions. 
\end{itemize}

Collectively, these contributions combine into a powerful and fully differentiable model for stripe pattern materials. In combination with adjoint sensitivity-analysis, Differentiable Stripe Patterns enable gradient-based optimization of high-level design objectives with opportunities for broad applications. We illustrate our method on a set of inverse design examples, including fabrics with 3D-printed reinforcements, structural textures for thin-walled 3D-prints, a compliant shell gripper, and personalized insoles.
We validate our examples through physical prototypes and observe good agreement between simulation and real-world behavior. 





\section{Related work}

\paragraph{Structured Materials}
Designing structured materials with desired macromechanical behaviors is an active field of research in material science and engineering \cite{bertoldi2017flexible}. 
With the widespread availability of 3D printing, the graphics community has likewise started to investigate the creation of metamaterials with lattice- \cite{panetta2015elastic,panetta2017worst,Gongora22Designing}, voxel- \cite{Bickel10Design,schumacher2015microstructures,zhu2017two}, and foam-like \cite{martinez2016procedural,martinez2017orthotropic} microstructures. A particular line of research focuses on two-dimensional, sheet-like materials \cite{Schumacher18Mechanical,martinez2019star,tozoni2020low,Leimer2020,Malomo18Flexmaps}.
Our work also targets sheet and thin shell materials, but focuses on global design optimization instead of periodically tileable material cells.

Closest to our setting in the context of material design is arguably the work by Tricard et al. \shortcite{tricard2020freely} for creating microstructures with freely-orientable channels. These thin-walled tubes are generated such as to follow a user-provided orientation field using a stochastic process. The resulting structures are extremely anisotropic, showing high stiffness along channels but compliance in the orthogonal plane. Our differentiable stripe patterns similarly target materials with high stiffness contrast along locally controllable directions. However, while Tricard et al. describe a forward geometry design process that is not informed by simulation, we propose an inverse material design approach in which stripe patterns are generated automatically such as to obtain desired mechanical behaviors.

Although not a primary focus of our work, we show through examples that Differentiable Stripe Patterns can be used to modulate the mechanical properties of conventional textiles using a print-on-fabric process similar to \cite{Perez:2017:CDA,Jourdan2021Printing}. Instead of using 3D-printing, an alternative strategy would be to use embroidery for reinforcement, as demonstrated by Moore et al. \shortcite{Moore2018} and Sati et al. \shortcite{Sati21Digisew}.


\paragraph{Eigenvector Optimization}
Since stripe patterns emerge as the solution to a generalized eigenvalue problem, optimizing for eigenvectors is at the core of our method. Designing mechanical systems with desired spectral properties is a problem that has occurred frequently in graphics literature. 
Applications include metallophones that produce desired sounds \cite{umetani10Designing,Bharaj15Computational,Musialski:2016:NSO},
coarse-level simulations \cite{Chen17Dynamics,Chen19EigenFit} and differential geometry operators \cite{Liu19Spectral,Chen20Spectral}  that preserve spectral properties, 
as well as mechanical assemblies that are robust to perturbations \cite{Thomaszewski14Computational,liu2022rigidity}.
A specific challenge in our setting is that all eigenvalues have nontrivial geometric multiplicity, making eigenvectors non-unique and derivatives undefined. By identifying the origins of this inherent multiplicity, we develop a strategy to obtain unique reference vectors in arbitrary eigenplanes and express eigenvectors in this space through an additional rotation parameter.



\paragraph{Designing Physical Surfaces}
As opposed to purely geometric design, physical surfaces are subject to equilibrium constraints. Applications include inflatable membranes \cite{Skouras14DIS,Panetta2021}, surfaces made from elastic curves \cite{Panetta19XShells,Zehnder16Ornamental,Pillwein2021,Neveu22Stability}, as well as auxetic \cite{KonakovicLukovic2018,Chen2021,JIANG2022103146} and tension-actuated \cite{Perez:2017:CDA,Guseinov17Curveups,Jourdan2022Computational} deployable structures. 
Adding to this line of work, we identify stripe patterns as a powerful paradigm for physical surfaces and present a differentiable pipeline for inverse design in that space.


\begin{figure*}[h]
    \centering
    \includegraphics[width = 1.0 \linewidth]{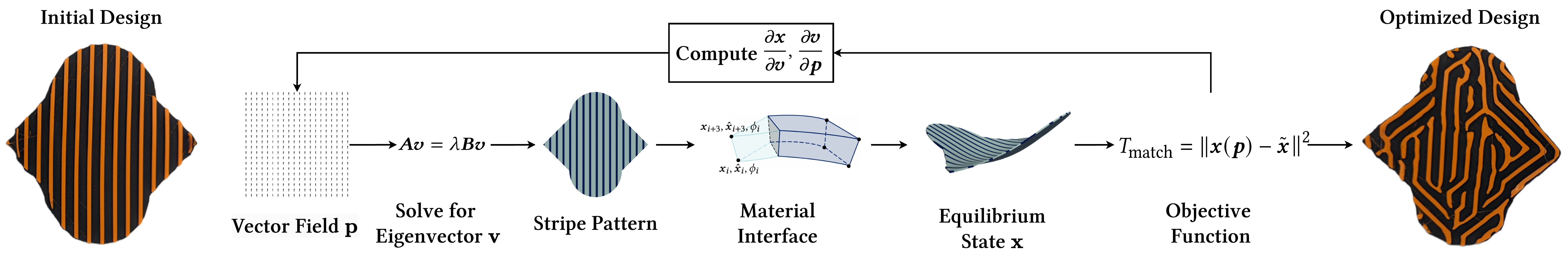}
    \caption{Overview of our differentiable stripe pattern pipeline. Starting from an initial input vector field, we compute the corresponding stripe pattern by solving a generalized eigenvalue problem (Sec. \ref{sec:stripePatterns}). To accurately model the mechanics of bi-material surfaces structured with such patterns, we combine solid shell finite elements with XFEM (Sec. \ref{sec:computationalModel}). 
    To enable gradient-based minimization of high-level design objectives,  we must compute derivatives of eigenvectors (Sec. \ref{sec:DSS}), material interface, and equilibrium states with respect to vector field parameters $\bp$. These design gradients are then used to automatically compute stripe patterns that best approximate macromechanical performance goals.
    }
    \label{fig:pipeline}
\end{figure*}

\paragraph{Material Interfaces and XFEM}
Our method uses stripe patterns to define bi-material distributions. This leads to the question of how to best model material interfaces that, in general, do not align with element boundaries in the simulation mesh.
XFEM is a technique for modeling material interfaces with displacement or strain discontinuities without mesh adaptation \cite{moes2003computational}. In the graphics literature, XFEM has been used, e.g., for cutting shells \cite{Kaufmann09Enrichment} and deformable solids \cite{Koschier2017}, metamaterial design \cite{zehnder2017metasilicone}, and for streamlined simulation of CAD models \cite{Hafner:2019}.
To unlock stripe patterns as a design space for structured surfaces, we seek an efficient shell model that is amenable to XFEM. While the method by Kaufmann et al. \shortcite{Kaufmann09Enrichment} is an option in principle, the high polynomial order of the bi-cubic patches together with the texture-based enrichment approach and discontinuous Galerkin formulation translate into excessive complexity for inverse design. While simpler discrete shell models from the graphics community \cite{grinspun2003discrete,Bridson02Simulation,Garg07Cubic,Chen18Physical}  offer favorable trade-offs between accuracy and computation cost, their discrete nature cannot be combined with standard XFEM approaches. We therefore propose a solid shell model \cite{Hauptmann98SolidShells,Ko17A6Node} that uses bi-linear triangular prism elements which, as we show in our analysis, offer good accuracy-performance trade-offs for the range of problems that we consider while integrating seamlessly with XFEM.


\paragraph{Stripe Patterns in Material Design}
Stripe patterns have recently seen increasing attention for material design and structural optimization. 
For example, Boddeti et al. \shortcite{Boddeti2020} use stripe patterns to guide the design of continuous reinforcement fields within a fiber-in-matrix approach.
Panetta et al. \shortcite{Panetta2021} proposed an approach for designing inflatables that deploy into desired shapes upon pressurization. Stripe patterns are used for initializing air channels, which are then optimized such as to best approximate a given target shape at equilibrium. 
Jourdan et al. \shortcite{Jourdan2022Computational} likewise use stripe patterns for initialization when designing reinforcement curves for self-deploying fabric models.
We share the excitement of these works for stripe patterns as a material design space. Instead of mere initialization, however, we propose a fully differentiable end-to-end pipeline enabling gradient-based design optimization within the space of stripe patterns.

\section{Overview}
The goal of our method is to automatically compute stripe patterns that, when used as bi-material distributions, lead to desired mechanical behavior. A visual summary of our method is shown in Fig. \ref{fig:pipeline}. To evaluate the performance of a design, we must compute its deformed equilibrium state given applied loads and boundary conditions using simulation. Determining the design parameters that best approximate a desired behavior requires inverting the entire design pipeline from input vector field to high-level mechanical function. This inversion notably includes derivatives of stripe patterns with respect to input vector fields, and derivatives of the equilibrium states with respect to material interfaces.

We start with a minimal description of Stripe Patterns \cite{Knoeppel15Stripe} in Sec. \ref{sec:stripePatterns}, including an analysis of eigenvalue multiplicity. To use Stripe Patterns for inverse design of bi-material distributions, we propose a solid shell model with extended finite elements to allow for moving material interfaces (Sec. \ref{sec:computationalModel}). Building on this computational model, we introduce Differentiable Stripe Patterns in Sec. \ref{sec:DSS}, including eigenvector derivatives and design objectives. 
We present design examples and further analysis in Sec. \ref{sec:results}.

\newcommand{\dprime}{{d^{\prime}}}
\newcommand{\dhat}{\hat{d}}

\section{Stripe Patterns}
\label{sec:stripePatterns}
The method by Kn{\"o}ppel et al. \shortcite{Knoeppel15Stripe} generates globally continuous, equally-spaced stripe patterns on arbitrary triangle meshes. It accepts as input a triangle mesh with $n$ vertices $\bx=(\bx_1, \ldots, \bx_n)$ as well as an initial vector field $\bz=(\bz_1, \ldots,\bz_n)$ with each $\bz_i \in \mathbb{R}^3$ indicating the vector field evaluated at vertex $\bx_i$. 
The goal of Stripe Patterns is to find an angle $\alpha_i$ per vertex that indicates its phase, i.e., the location in a periodic function that determines which material will be assigned to each point. To make this angle function globally periodic and smooth, angles are represented as complex numbers $\Psi_i\in\mathbb{C}$ with $\alpha_i=\arg \Psi_i$. Stripes should be locally orthogonal to the input vector field, i.e., the gradient of the angle function should be collinear to the input vector field\Changes{;} $\nabla_{(u,v)} \alpha(u,v)= \bz(u,v)$ where $(u,v)$ are local coordinates.
In the discrete setting, changes in angle are measured along edges $\be_{ij}$ of the triangulation by linear interpolation and integration as
\begin{equation}
    \omega_{ij}= \frac12(\be_{ij}^T\bz_i+\be_{ij}^T\bz_j) \ .
\end{equation}
In general, the change in angle cannot perfectly agree with the vector field, unless $\bz$ is integrable. For this reason, the mismatch between target change and actual change along all edges is minimized as 
\begin{equation}
\label{eq:mismatchPhaseEnergy}
    \Changes{
    E_\Psi=\sum_{ij\in \mathcal{E}}w_{ij}|\Psi_j-e^{\iota\omega_{ij}}\Psi_i|^2 ,}
\end{equation}
where $\mathcal{E}$ denotes the edge index set and $w_{ij}$ is a mesh-dependent weight. To arrive at an algorithm with only real arithmetics, per-vertex phases are expanded into real and imaginary parts as \Changes{$\Psi_i=a_i+\iota b_i\rightarrow(a_i,b_i)$} which we store in a real-valued vector \Changes{$\eigvec\in\mathbb{R}^{2n}$}. Furthermore, since the above energy is quadratic in phases, the real-valued analogue can be expressed as \Changes{$\frac12 \eigvec^T\bA\eigvec$} with a matrix $\bA$ that depends on the input mesh and vector field.
To eliminate the trivial minimizer \Changes{$\eigvec=0$}, Kn{\"o}ppel et al. pose the \Changes{constrained} optimization problem 
\begin{equation}
\label{eq:EVPOptProblem}
         \Changes{\min_{\eigvec} \frac12 \eigvec^T\bA\eigvec \quad \text{s.t.} \quad \frac12\eigvec^T\bB \eigvec = 1 \ ,}
\end{equation}
where $\bB$ is a \Changes{lumped} mass matrix whose per-vertex entries are computed by summing up areas of incident triangles. Considering the first part of the first-order optimality conditions
\begin{equation}
\label{eq:generalizedEVP}
    \Changes{
    \bA\eigvec-\lambda\bB\eigvec = \bzero \ , 
    }
\end{equation}
it is clear that solutions to this \Changes{constrained} optimization problem are generalized eigenvectors of $\bA$.
Consequently, eigenvectors corresponding to the smallest eigenvalue will minimize the mismatch energy across all eigenvectors of $\bA$. This vector is then further processed through component-wise normalization, nonlinear per-triangle interpolation, and other filtering operations to yield per-vertex \textit{texture coordinates} that ultimately determine the \textit{color}---or, in our case, material identity---of each point on the mesh. 
We refer to this process as the \textit{forward stripe problem}. For our method, we aim to compute stripe patterns that lead to desired mechanical performance. Solving this \textit{inverse stripe problem} involves computing derivatives of eigenvectors of $\bA$ with respect to the input vector field, which is made challenging due to the following observation.
\begin{theorem}
All generalized eigenvalues of $\bA$ have geometric multiplicity two, i.e., all eigenvalues are duplicate and the dimension of the associated eigenspace is two.
\end{theorem}
\begin{proof}
To see this, we observe that Eq. (\ref{eq:mismatchPhaseEnergy}) measures only differences between phases along edges---rotating all per-vertex phases by the same angle will not change the energy. Consider a given eigenvector $\bv$ and its corresponding phase vector $\Psi$. Let $\Psi^\perp$ denote the phase vector defined as \Changes{$ \Psi_j^\perp= \iota\Psi_j$}, i.e., with per-vertex phases rotated by $\pi/2$. The corresponding real-valued vector is $\bv^\perp$ with $\bv^\perp_j=(-b_j,a_j)$. Since $\bv^\perp$ has the same norm and energy as $\bv$, it must also be an eigenvector for the same eigenvalue.
\end{proof}

\begin{figure}[t]
    \centering
    \includegraphics[scale=1.0]{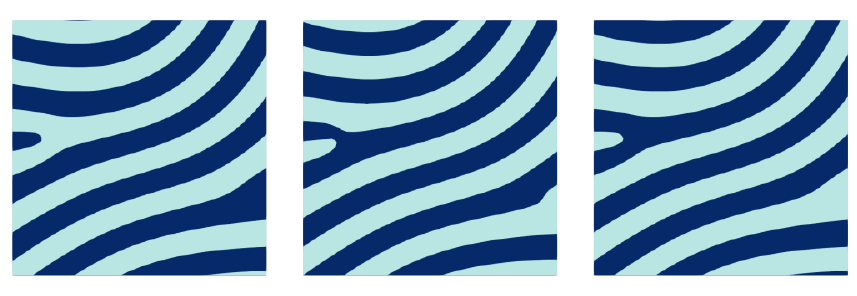}
    \caption{Sequence of stripe patterns corresponding to eigenvectors from a two-dimensional eigenspace. Linear motion in this eigenspace corresponds to coordinated phase changes that yield travelling waves.}
    \label{fig:eigenvectorNullspace}
\end{figure}

The consequence of this observation is that eigenvectors are not unique and derivatives not well-defined. Fig. \ref{fig:eigenvectorNullspace} illustrate the space of eigenvectors for a given eigenplane; see also the accompanying video for an animation. It can be seen that motion in this eigenspace corresponds to globally coordinated phase changes, yielding waves travelling across the surface. To compute derivatives of stripe patterns with respect to the input vector field, we must make eigenvectors unique and their derivatives well-defined. Before we describe our solution to this problem in Sec. \ref{sec:DSS}, we first introduce our computational model.

FF\section{Solid Shells and XFEM}
\label{sec:computationalModel}

Stripe patterns give rise to piece-wise linear material interfaces which, in general, will cross through elements of the underlying triangulation. To predict the mechanical behavior of a given stripe pattern, interfaces between soft and rigid materials must be modeled accurately. 
Since we are interested in bi-material distributions with high stiffness ratio, simple interpolation of material parameters is not an option \cite{kharevych2009numerical}.
Conforming meshes would avoid material interpolation but changing interfaces require remeshing, which is not differentiable. 
For this reason, we opt for an extended finite element method (XFEM) that models material interfaces with additional degrees of freedom \cite{moes2003computational}. As its central benefit, this approach allows for continuous evolution of material interfaces without changes to the underlying mesh. The additional degrees of freedom, so called enrichment coordinates, allow for strain discontinuities across the interface, which is precisely what is needed in our context.

\subsection{Solid Shells}
To model structured surfaces that can stretch and bend, we combine XFEM with solid shell finite elements---volumetric elements in the shape of triangular prisms \cite{Ko17A6Node}.
Given a triangle mesh as input, we start by creating two offset surfaces by extruding mid-surface vertices along their vertex normals. The resulting extruded surfaces model the top and bottom layers of the structured surface and are used to define a set of six-node triangular prism elements as shown in Fig. \ref{fig:prismElement}.
We use a standard isoparametric approach and define the deformed geometry for each element as 
\begin{equation}
    \bx(\bu) = \sum_{i=1}^6 \bx_iN_i(\bu) \ ,
\end{equation}
and analogously for the rest state geometry $\bX$.
In these expressions, $\bu=(u,v,t)$ are local coordinates with $(u,v)$ parametrizing the mid-surface triangle and the third coordinate $t\in[-h/2,h/2]$ runs along the normal direction of the shell, with $h$ indicating its thickness. The six bi-linear basis functions $N_i$ are uniquely defined by the interpolation property $ \Changes{ N_i(\bu_j) = \delta_{ij}}$. It is worth noting that, for any fixed $t$, the in-plane coordinates define a linear triangle element with constant strain. Even though deformation is constant for any given thickness value, it varies linearly through the thickness. It is this linear variation that, as we will see, allows solid shells to perform significantly better for bending deformation than constant-strain linear tetrahedra.

\begin{figure}[h]
    \centering
    \includegraphics[scale=1.0]{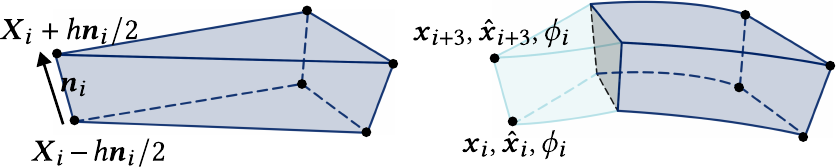}
    \caption{Solid Shell Elements. \textit{Left}: a bi-linear prism element is generated by extruding a triangle of the input mesh in the direction of its vertex normals. \textit{Right}: level set values $\phi_i$ define the location of a material interface and additional enrichment DOFs $\hat{\bx}_i$ allow for different strains on the two sides. 
    }
    \label{fig:prismElement}
\end{figure}

Using the continuous approximations of element geometry, all other kinematic quantities are obtained in the standard form, i.e.,
\begin{equation}
    \bF(\bu)=\frac{\partial \bx}{\partial \bX}|_{(u,v,t)} \, \quad \bC(\bu)=\bF(\bu)^T\bF(\bu) \ ,
\end{equation}
where we made explicit the fact that, unlike for linear elements, the deformation gradient $\bF$ and the right Cauchy-Green tensor $\bC$ vary across the element.
 
\paragraph{Strain Energy}
We use a standard Neo-Hookean constitutive law as the basis for both soft and stiff materials. The energy density of this material is defined as 
\begin{equation}
    \Psi_\mathrm{Iso}=\frac{1}{2}\left[
    \mu(\text{tr}(\bC)-3)
    -2\mu\ln J
    +\lambda(\ln J)^2 
    \right]
    \ ,
\end{equation}
where $\lambda$ and $\mu$ are the Lam{\'e} parameters.
We use different application-dependent material parameters for soft and stiff materials, see Sec. \ref{sec:results}. If orthotropic behavior is desired, the isotropic base material can be augmented with stiffening fibers as
\begin{equation}
    \Psi_\mathrm{Ortho}=\Psi_\mathrm{Iso} + \frac{1}{2}\beta_f\bn_f^T\bC\bn_f   \ ,
\end{equation}
where $\bn_f$ is the unit vector indicating the fiber orientation and $\beta_f$ is a material parameter that reflects the fiber stiffness and density; see, e.g., \cite{Holzapfel01AViscoelastic}.


We then obtain the elastic energy of a given solid shell element by integrating the strain energy density across the element. In practice, we approximate this integral using numerical quadrature, 
\begin{equation}
U_e=\int_\Omega \Psi(\bC(\bu)) \ dV \approx \sum_j w_j\Psi(\bC(\bq_j)) \left|\frac{\partial \bX}{\partial \bq} \right|\ ,
\end{equation}
where $\bq_j=(u_j,v_j,t_j)$ are quadrature points in generic coordinates, and $w_j$ are corresponding quadrature weights. For elements with a single material, we use a six-point numerical quadrature scheme with three points per thickness value. 
Elements crossed by a material interface require more elaborate treatment, as explained next.

\subsection{Material Interfaces \& XFEM}

\paragraph{Level Sets for Material Interfaces}
To model the bi-material distributions induced by stripe patterns, we must determine the location of the material interfaces.
We model these interfaces using a level set $\levelset$ defined through a set nodal values $\levelset_i\in\mathbb{R}$. The nodal values represent the signed distance to the interface, with the sign indicating whether the node is on the soft or stiff side.
To convert a given stripe pattern to its corresponding level set,  we first retrieve the per-vertex angle $\alpha_i = \mathrm{arg}(\phases_i)$ from the corresponding eigenvector component $\phases_i$. We then compute level set values using a smoothed triangle wave as transfer function,
\begin{equation}
    \label{eq:levelset}
    \Changes{
    \phi_i=1-\frac{2\mathrm{arccos}[(1-a_1)\mathrm{sin}( \alpha_i-\frac \pi 2)]}{\pi}-a_2,}
\end{equation}
where $a_1$ is a smoothing term and $a_2$ determines the ratio between stiff and soft material by vertically translating the cut-off value, see Fig. \ref{fig:triangleWaveFunction}.

\begin{figure}[h]
    \centering
    \includegraphics[scale=1.0]{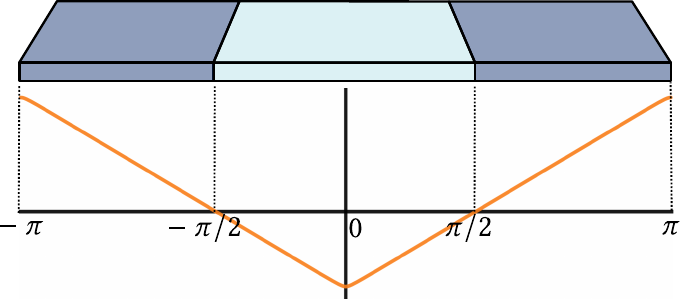}
    \caption{Quasi-linear transfer function for mapping per-vertex angles $\alpha_i$ into level set values $\phi_i$.  Smooth junctions between linear segments ensure differentiability. The material interface is located at $\phi=0$,  $\phi>0$ corresponds to stiff materials, whereas $\phi<0$ indicates soft material. The ratio between stiff and soft material can be controlled by translating this curve vertically. }
    \label{fig:triangleWaveFunction}
\end{figure}

\paragraph{XFEM}
Elements whose nodes have level set values $\levelset_i$ with different signs are crossed by a material interface. To accurately model the behavior of these bi-material elements we resort to an extended finite element approach. While the location of the interface is given by the level set values $\levelset_i$, allowing for different strains on opposite sides of the interface requires additional degrees of freedom, the so called \textit{enrichment coordinates} $\enriched$. Enrichment coordinates $\enriched_i$ are collocated with the standard nodal DOFs $\bx_i$ but have special basis functions.

\begin{figure}[h]
    \centering
    \includegraphics[scale=1.0]{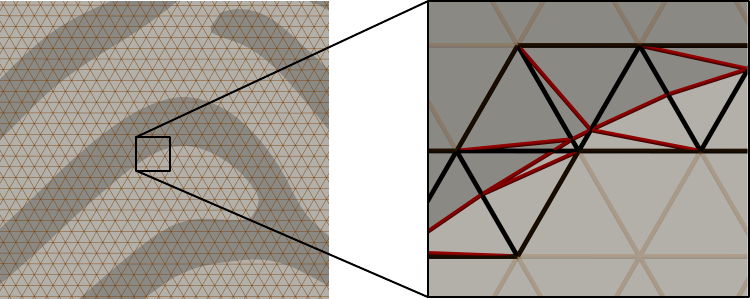}
    \caption{\Changes{XFEM discretization of stripe patterns with regions of stiff (\textit{dark gray}) and soft (\textit{light gray}) material as indicated. Elements crossed by an interface (\textit{black}) are split into three sub-elements (\textit{red}) to integrate their elastic energy.}  }
    \label{fig:xfemDiscretization}
\end{figure}

We choose the ridge function \cite{moes2003computational}, which was designed for material interfaces that are continuous in displacements but discontinuous in strain ($C^0$ but not $C^1$). 
For simplicity, we only allow one interface per element, 
although generalizations are possible \cite{zehnder2017metasilicone}. This can be achieved by adjusting the frequency of the stripe pattern for a given input mesh, and vice-versa.
The ridge enrichment function is given by
\begin{equation}
    \psi=\sum_{i}|N_i\levelset_{i}|-|\sum_{i}N_i\levelset_{i}| \ ,
\end{equation}
with the corresponding interpolation function for the nodal positions
\begin{equation}
    \deformed=\sum_{i}{N_i(\bu)\deformed_{i}}+\psi\sum_{i}{N_i(\bu)\enriched_{i}} \ .
\end{equation}
We refer to Zehnder et al. \shortcite{zehnder2017metasilicone} for a visual construction of these functions.
The strain in the enriched elements varies throughout the element, with a discontinuity located at the interface. To properly integrate over these enriched elements, we subdivide them into three sub-prisms and evaluate the sub-integrals numerically using an 11-point quadrature rule \Changes{(Fig. \ref{fig:xfemDiscretization})}. 

\subsection{Evaluation}
Solid shells have, to the best of our knowledge, not been explored in graphics literature before. 
We therefore perform a series of experiments intended to test the accuracy and convergence behavior of solid shells in comparison to other models. In particular, we compare to conforming discretizations with four-node linear tetrahedron elements and ten-node quadratic tetrahedron elements, as well as conforming bi-linear prisms and a discrete shell model \Changes{\cite{gingold2004discrete}}.
We use two simple experiments in which we impose different periodic boundary conditions---cylindrical bending (radius $r=10cm$) and uni-axial loading (stretch $\varepsilon=10\%$)---onto a square plate with dimensions $7cm\times7cm$. 
\begin{figure}[h]
    \centering
    \includegraphics[scale=1.0]{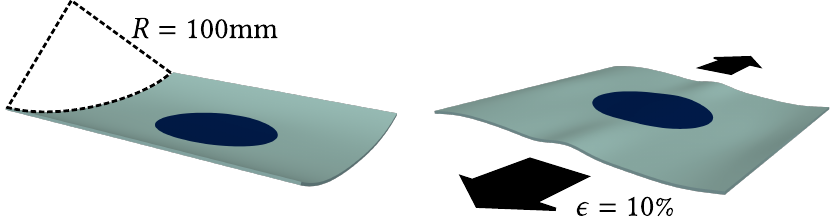}
    \caption{Problem setups for shell comparison. A square plate is endowed with a disc-shaped region of stiff material (\textit{dark blue}) subject to periodic boundary conditions that impose cylindrical curvature (\textit{left}) and uni-axial stretching (\textit{right}) as indicated.}
    \label{fig:solidShellExperiments}
\end{figure}

\begin{figure}[b]
    \centering
    \includegraphics[scale=1.0]{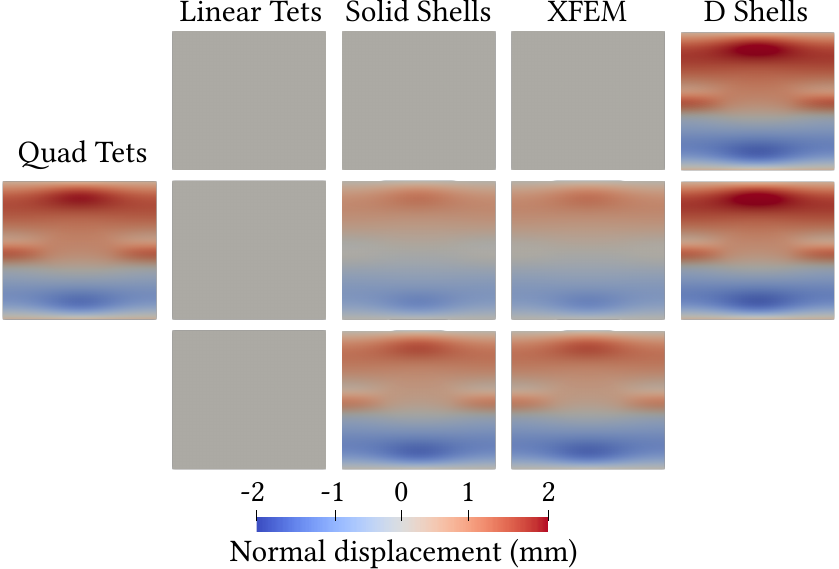}
    \caption{Normal displacement plots for uni-axial loading.  Almost all models produce out-of-plane buckling given sufficient resolution, whereas linear tetrahedra completely fail to capture this behavior.}
    \label{fig:shellCompBending}
\end{figure}
To investigate the impact of material interfaces, we use a disc-shaped region with stiffer material at the center of the plate. See also Fig. \ref{fig:solidShellExperiments}. We solve these problems numerically using different types of elements, different mesh resolutions, and plates of different thickness. We then compare the elastic energy and per-vertex displacement norms at equilibrium. 

The results for the bending test are summarized in Tab. \ref{tab:timePerformance}. Quadratic tetrahedra perform best and converge very rapidly under refinement. Because differences are very small, we only list energy and compute time for the coarsest discretization and use them as reference values.
It can further be noted that  \Changes{the discrete shell model \cite{gingold2004discrete}} performs extraordinarily well, offering the best compromise between accuracy and computation time.
The picture is very different for linear tetrahedra, which perform poorly and worst by far. Solid shell elements fare much more favorably in comparison. Especially the moderate resolution version seems to offer an interesting trade-off between accuracy and performance. Tab. \ref{tab:timePerformance} also shows that the performance of solid shell elements and linear tetrahedra degrades as the shell thickness decreases from $0.6mm$ to $0.15mm$. While these observations imply that solid shells are not competitive for extremely thin shells, we are primarily interested in problems with thickness values above $1mm$. 

The results for uni-axial stretching underscore these initial impressions.
As illustrated in Fig. \ref{fig:solidShellExperiments}, the stiff disk prevents homogeneous stetching which, in combination with a nonzero Poisson ratio, leads to spatially varying compressive stresses in the orthogonal direction that induce buckling.
Fig. \ref{fig:shellCompBending} shows results for this experiment in the form of normal displacement plots. Using again quadratic tetrahedra as ground truth, it can be seen from the plots that results for the discrete shell model are quite close to the reference solution, although convergence under refinement is somewhat unclear. Solid shells do not produce satisfying results for the lowest resolution, but converge to qualitatively acceptable solutions under refinement. In stark contrast, linear tetrahedra fail to produce out of plane motion altogether, which we attribute to the well-documented phenomenon of volumetric locking. 

\begin{table}[t]
	\centering
	\caption{Comparison between different shell models on a cylindrical bending example. }
        \Changes{
        \begin{tabular}{p{2.0cm}p{1.1cm}p{1.1cm}p{1.1cm}p{1.0cm}}
		\toprule
		Method & DOFs & Energy (0.6mm) & Energy (0.15mm)  & Timing [s]\\
		\midrule
            Quadratic Tets & 73764 & 0.0596 & 0.0009  & 10.86\\
            \midrule
		Linear Tets & 12381 & 0.3542 & 0.0483  &  0.74 \\
		          &47914 & 0.2072 & 0.0146  &  3.61 \\
		           & 244285 & 0.0951 & 0.0070  &  35.36 \\
            \midrule
		Solid Shells & 11884 & 0.1215 & 0.0164  &  0.70 \\
		           & 46346 & 0.0751 & 0.0048  & 3.50 \\
		           & 183022 & 0.0635 & 0.0032  &  17.96 \\
            \midrule
		XFEM       & 12126   & 0.1210 & 0.0161  &  0.83 \\
		            & 46864 & 0.0752 & 0.0048 &  3.90 \\
		           & 184438 & 0.0635 & 0.0032  &  19.84 \\
            \midrule
            Discrete Shells & 5942 & 0.0598 & 0.0009  &  0.56 \\
		              & 23173 & 0.0597 & 0.0009  &  2.60 \\
		\bottomrule
	\end{tabular}
 }
	\label{tab:timePerformance}
\end{table}

\paragraph{Discussion}
Our analysis shows that the discrete shell model offers very good accuracy already for comparatively coarse discretizations---on par with quadratic tetrahedra, but at a much lower computational costs. While this makes them an ideal choice for static bi-material distributions, we were unable to find a robust formulation for moving interfaces with conforming meshes. Our experiments with different options invariably led to numerical instabilities: directly moving interface nodes along with the stripe pattern produces ill-shaped elements and diverging derivatives. Moving non-interface nodes using, e.g., a Laplacian regularizer \cite{Perez:2017:CDA,Montes20Skintight} delays these problems without solving them. While remeshing can avoid instabilities in principle, the corresponding discontinuities would prevent gradient-based optimizers from converging. 
These observations are in line with expectations as they are among the primary reasons for using XFEM instead of conforming discretizations in inverse problems with material interfaces \cite{moes2003computational,zehnder2017metasilicone,Kaufmann09Enrichment}. While discrete shells and XFEM would appear to be a winning combination, we could not find a way to reconcile these disparate  concepts: extended finite elements rely on element geometry and strain being defined in terms of basis functions and their derivatives, but the \Changes{curvature computation of discrete shell elements} does not fit this framework. 
Ultimately, solid shells with XFEM emerge as the best trade-off for our problem setting. Unlike linear tetrahedra, they provide acceptable accuracy already for moderate resolution. Compared to quadratic tetrahedra with XFEM, solid shells are considerably simpler and, for the level of accuracy and mesh resolution that we require, the more efficient choice. 

\section{Differentiable Stripe Patterns}
\label{sec:DSS}
Inverse material design based on stripe patterns leads to the central challenge of determining the change in stripes induced by a given change in input vector field $\bz$. Since we only care about the direction of per-vertex vectors $\bz_i$ and not their magnitude, we parameterize the input vector field as $\bz=\bz(\bp)$ where $\bp=(p_1,\ldots,p_n)$ is the vector of design parameters with per-vertex angles $p_i$.
The relation between design parameters and output material distribution is then given implicitly through the generalized eigenvalue problem (\ref{eq:generalizedEVP}). However, since eigenvalues for this problem have geometric multiplicity two, the corresponding eigenvectors are not unique. 
We first show that eigenvector derivatives are not well defined before we describe our solution to this problem.


\subsection{Eigenvector Derivatives}
\label{sec:eigenDerivatives}

We start by expressing optimization problem (\ref{eq:EVPOptProblem}) through its Lagrangian
\begin{equation}
\Changes{
    \mathcal{L}(\eigvec,\eigval)=\frac12\eigvec^T\bA\eigvec -\frac\eigval2\left(  \eigvec^T\bB\eigvec - 1\right)\ .}
\end{equation}
Any solution $\bs=(\eigvec,\eigval)$ to this problem must satisfy the first-order optimality conditions $\bg(\eigvec,\eigval):=\nabla\mathcal{L}=\mathbf{0}$ where
\begin{align}
    \bg_v = (\bA -\eigval \bB) \eigvec  \ , \\
    \Changes{
    \bg_\lambda = -\frac 12 \left( \eigvec^T\bB\eigvec -1\right) \ . }
\end{align}
Requiring  that these conditions be satisfied for all admissible parameter changes,
\begin{align}
    \frac{d \bg}{d\bp} = \frac{\partial \bg}{\partial \bp} + \frac{\partial \bg}{\partial \bs}\frac{\partial \bs}{\partial \bp}=\mathbf{0} \ ,
\end{align}
leads to the saddle-point system
\begin{equation}
    \begin{bmatrix}
    \bA-\eigval \bB & -\bB\eigvec \\
    -(\bB\eigvec)^T & 0 
    \end{bmatrix}
    \begin{bmatrix}
    \frac{\partial \eigvec}{\partial \bp} \\
    \frac{\partial \eigval}{\partial \bp}
    \end{bmatrix}
    =
    \begin{bmatrix}
    -\frac{\partial \bg_\eigvec}{\partial \bp} \\
    0
    \end{bmatrix}
    \ .
\end{equation}
Since $\eigval$ has geometric multiplicity two, the system is singular and, consequently, the derivative of $\eigvec$ is not well-defined. This is easily verified by observing that the vector $(\eigvec_\perp,0)$ is in the nullspace of $\dgds$, where $\eigvec_\perp$ is the second eigenvector for $\eigval$ with $\eigvec^T\eigvec_\perp=0$.  

To resolve eigenvector ambiguity, we recall from Sec. \ref{sec:stripePatterns} that vectors in the two-dimensional eigenspace correspond to synchronous rotations of per-vertex phase vectors. Using this insight, we can define unique eigenvectors by specifying the phase for an arbitrary vertex $k$. To this end, we simply set $\bv_k=(a_k,0)$ which eliminates one degree of freedom. This leads to a modified optimization problem with an extra constraint on $b_k$,
\begin{equation}
\label{eq:LagrangianRegularized} 
    \mathcal{L(\bp,\eigvec,\lambda,\mu)} = \frac{1}{2} \bv^T\bA\bv - \frac{1}{2}\lambda (\bv^T\bB\bv-1) - \mu b_k,
\end{equation}
whose corresponding first-order optimality conditions are
\begin{align}
    \bg_\eigvec & = \bA\eigvec - \lambda\bB\eigvec - \mu \be_{b_k} = 0  \label{eq:optCond2_1} \\
    \bg_\eigval &= -\frac{1}{2}(\bv^T\bB\bv-1) = 0 \label{eq:optCond2_2}\\
    \bg_\mu &= -b_k = 0 \label{eq:optCond2_3} \ .
\end{align}
A vector $\bv$ that satisfies these optimality conditions is also an eigenvector of $\bA-\lambda\bB$, since we can always find a vector in the eigenplane that satisfies the condition $b_k=0$. Consequently, $\mu$ is always equal to $0$. Requiring Eqs (\ref{eq:optCond2_1}---\ref{eq:optCond2_3}) to be satisfied for any change in parameters yields the modified saddle-point system
\begin{equation}
\label{eq:KKT2}
    \begin{bmatrix}
    \bA-\eigval \bB & -\bB\eigvec & \be_{b_k} \\
    -(\bB\eigvec)^T & 0 & 0 \\
    \be_{b_k}^T & 0 & 0
    \end{bmatrix}
    \begin{bmatrix}
    \frac{\partial \eigvec}{\partial \bp} \\
    \frac{\partial \eigval}{\partial \bp} \\
    0
    \end{bmatrix}
    =
    \begin{bmatrix}
    -\frac{\partial \bg_\eigvec}{\partial \bp} \\
    0 \\
    0
    \end{bmatrix}
    \ ,
\end{equation}
where $\be_{b_k}$ is the unit vector corresponding to the index $b_k$. Intuitively, there exists no vector from the eigenplane that is both orthogonal to $\eigvec$ and $\be_{b_k}$. The above system is therefore non-singular and can be solved for the unknown eigenvector derivatives $\frac{\partial \eigvec}{\partial \bp}$.

\paragraph{Eigenplane Parameterization}
It should be noted that, while constraining a given phase variable leads to unique eigenvectors, any specific choice of $k$ and $b_k$ will introduce bias---another vector from the eigenplane might be better suited for decreasing the design objective. To avoid arbitrary choices,  we introduce an additional parameter $\theta$ that represents the current eigenvector as a rotation in the eigenplane relative to the unique reference eigenvector defined through the choice of $k$ and $b_k$.

\subsection{Regularization and Parametrization}

\paragraph{Vanishing Phase Values}
The idea of stripe patterns is to express phase as the argument of complex number since this representation can offer global continuity. Once eigenvectors are computed, per-vertex phases (complex numbers) have to be normalized to retrieve their argument (i.e., the real-numbered angle value). However, as phase vector tend towards zero, this operation becomes numerically unstable and, eventually, diverges. This is typically not much of a concern for the forward stripe pattern problem, since even numerically infinitesimal vectors can still be normalized. However, the sensitivities of angles diverge as per-vertex phases approach zero, which is very problematic for the inverse problem.
Our solution to this problem is to simply disallow phase vectors from becoming arbitrarily small. We achieve this goal by imposing smoothly-clamped logarithmic barrier functions \cite{li2020incremental} on per-vertex phases as,
\begin{equation}
    R_\mathrm{sing}(d,-\hat{d})= 
    \begin{cases}
    -(d-\hat{d})^2 \ln (\frac{d}{\hat{d}}) & 0<d\leq\hat{d} \\
    0 & \text{otherwise} \ ,
    \end{cases}
\end{equation}
where $d=|\eigvec_i|$ and $\hat{d}=0.1$ is a cut-off value. 
By preventing zero phases through a design objective, we ensure that eigenvectors still have the same meaning as before---they represent ideal approximations to potentially non-integrable vector fields. However, the design parameters are now strongly repelled from values that lead to eigenvectors with zero per-vertex phases.

An interesting question in this context is \textit{what do we lose when preventing zero per-vertex phases}? As described by Kn{\"o}ppel et al. \shortcite{Knoeppel15Stripe}, scaling per-vertex phases down to zero permits the introduction of singularities in the phase field that would otherwise lead to non-integrability. Our experiments indicate that the expressiveness of stripe patterns generated without such non-integrable components does not suffer. Indeed, while almost-zero phase vectors can be normalized without problems, they often lead to randomly oriented phases with significant curl. These high-curl regions translate into quasi-amorphous material distributions with isolated patches rather than clean stripes. We argue that such phase fields are neither useful nor desirable from a material design perspective.

\paragraph{Smoothness}

While the regularizer introduced above prevents singularities in the phase field, the input vector field can still exhibit high-frequency components that translate into artifacts in the stripe patterns as shown in Fig. \ref{fig:smoothRegularizer}.
To avoid such artifacts, we introduce an additional smoothness regularizer based on the formulation by Crane et al. \shortcite{Crane:2010:TCD},
\begin{equation*}
    R_\mathrm{sm}(\bp) = \sum_{i,j\in \mathcal{E}} w_{ij}
    \left([\mathrm{cos}(p_i)-\mathrm{cos}(p_j)]^2 + [\mathrm{sin}(p_i)-\mathrm{sin}(p_j)]^2 
    \right) \ , \\
\end{equation*}
where $\mathcal{E}$ denotes the edge index set, $w_{ij}$ is  cotangent weight, and $p_i$, $p_j$ are design parameters (i.e., angles) expressed in a common reference frame. 
As can be seen from Fig. \ref{fig:smoothRegularizer} (\textit{right}), this simple regularizer leads to substantially cleaner stripe patterns. It furthermore reduces singularities even before the log barrier penalty becomes active. 

\begin{figure}[h]
    \centering
    \includegraphics[scale=1.0]{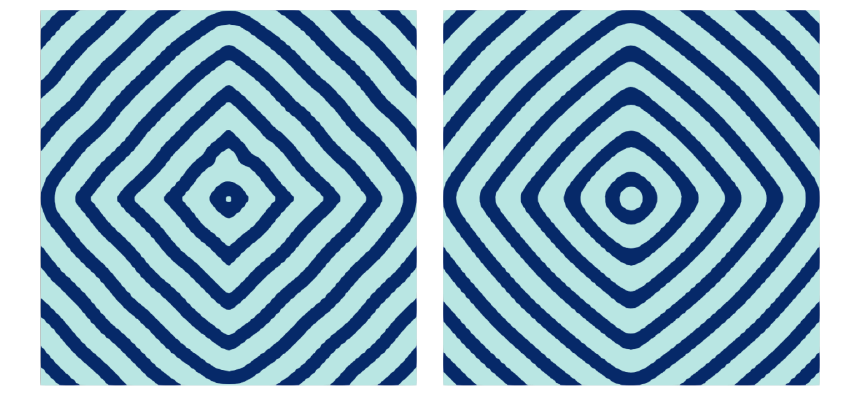}
    \caption{Smoothness regularizer. By promoting smooth vector fields during optimization, this regularizer removes high-frequency noise (\textit{left}), resulting in clean and aesthetically pleasing stripe patterns (\textit{right}).
    }
    \label{fig:smoothRegularizer}
\end{figure}


\begin{figure*}[h]
    \centering
    \includegraphics[width=\textwidth]{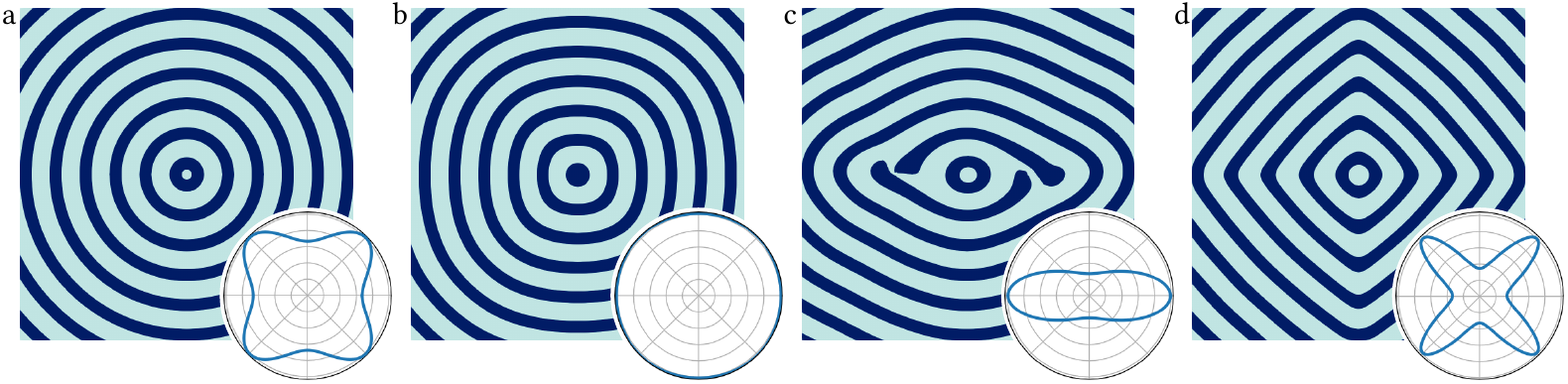}
    \caption{Controlling Macromechanical Properties. Starting from an initial concentric pattern (\textit{a}), we optimize for patterns that lead to isotropic (\textit{b}), orthotropic (\textit{c}), and tetragonal (\textit{d}) macromechanical behavior. }
    \label{fig:FabricModulation}
\end{figure*}

\subsection{Design Objectives}
Our differentiable formulation for stripe patterns can be used for solving a large array of design problems, see Sec. \ref{sec:results}. In the following, we summarize the objective functions that we use for our examples.

\paragraph{Target Deformation}
A basic task is to find stripe patterns such that the design best approximates a given target shape $\tilde{\bx}$ in equilibrium. In its simplest form, the corresponding objective function is
\begin{equation}
\label{eq:targetDeformation}
    T_\mathrm{match} = |\bx(\bp)-\tilde{\bx}|^2 \ .
\end{equation}

\paragraph{Macromechanical Properties}
Stripe patterns are ideal for creating materials with high stiffness contrast, e.g., by modulating the macromechanical properties of a soft (fabric) substrate with stripe patterns made from stiffer material. 
In this context, designers are typically interested in averaged, high-level behavior of the reinforced material, rather than in local deformations. In order to find stripe patterns that best approximate a given macromechanical behavior, we define an objective function that penalizes deviations between actual and target directional stiffness as 
\begin{equation}
    T_\mathrm{Mat} = \sum_i |k(\theta_i,\bx(\bp))-\hat{k}_i|^2 \ ,
\end{equation}
where $\theta_i\in[0,\pi]$ are sample locations in polar coordinates with corresponding directional stiffness targets $\Changes{ \hat{k_i}}$ given as generalized Young's moduli. To compute the directional stiffness for given design parameters, we simulate a unit-patch of the corresponding structure subject to periodic boundary conditions that enforce a given uniaxial target deformation. We obtain the corresponding stiffness $k(\theta_i,\bx(\bp))$ by evaluating the stress on the boundary of the unit cell and dividing by the strain magnitude as described by Schumacher et al. \shortcite{Schumacher18Mechanical}.

\paragraph{Generalized Stiffness}
Another basic design goal that we use in our examples is to achieve a desired stiffness with respect to given loads. When defining generalized stiffness as displacement norm divided by applied force magnitude, this objective becomes a special case of Eq.  (\ref{eq:targetDeformation}) . 

\subsection{Optimization Algorithm}
Our formulation allows us to compute derivatives of objective functions with respect to the design parameters. These gradients can then be used to drive first-order descent or quasi-Newton methods for minimization. We experimented with steepest descent and L-BFGS, but ultimately found the Globally-Convergent Method of Moving Asymptotes (GC-MMA) \cite{svanberg1995globally} to be most efficient for our problem setting.
\Changes{We express the objective gradient with respect to the parameters as
\begin{equation}
\label{eq:dTdp}
\begin{split}
    \frac{dT}{d\bp} = \frac{\partial T}{\partial \bp} + \left(\frac{\partial \boldsymbol{\phi}}{\partial \bv} \frac{\partial \bv}{\partial \bp}\right)^T \left(\frac{\partial T}{\partial \boldsymbol{\phi}} + \frac{\partial \bx}{\partial \boldsymbol{\phi}}^T \frac{\partial T }{\partial \bx}\right) \ ,
    \end{split}
\end{equation}}
\Changes{ where $\frac{\partial \phi}{\partial \bv}$ is the derivative of the level set with respect to the eigenvector (Eq. \ref{eq:levelset}), $\frac{\partial \bv}{\partial \bp}$ denotes the sensitivity of the eigenvector with respect to the parameters as illustrated in section \ref{sec:eigenDerivatives} and $\frac{\partial \bx}{\partial \boldsymbol{\phi}}$ is the sensitivity of the nodal positions with respect to the level set. Instead of fully evaluating both sensitivity matrices, we use the adjoint method for both terms sequentially, requiring just two linear solves per gradient evaluation.}

\Changes{\paragraph{Simulation Sensitivity}
The map of the nodal positions as a function of the level set $\bx(\boldsymbol{\phi})$ is given by the force equilibrium constraint of the forward simulation, $\boldsymbol{f}=-\frac{dU}{d\bx}=0$. As a result, changes in $\boldsymbol{\phi}$ induce changes in the equilibrium configuration $\bx$. Since $\frac{d\boldsymbol{f}}{d\boldsymbol{\phi}}=0$ for every equilibrium configuration, we have}
\begin{equation}
\Changes{
    \frac{d\boldsymbol{f}}{d\boldsymbol{\phi}}=\frac{\partial\boldsymbol{f}}{\partial\boldsymbol{\phi}}+\frac{\partial \bx}{\partial \boldsymbol{\phi}}^T\frac{\partial\boldsymbol{f}}{\partial\bx}=0\ ,
    }
\end{equation}
\Changes{from which $\frac{\partial \bx}{\partial \boldsymbol{\phi}}$ is extracted.}

\paragraph{Evaluating Candidate Parameters}
Whenever a candidate parameter update is evaluated during the optimization process, we must recompute eigenvectors at the new parameter location. 
To this end, we compute a pair of orthogonal eigenvectors $(\eigvec_1,\eigvec_2)$ as a basis for the eigenplane with minimal eigenvalue. We obtain the reference eigenvector by finding $\theta$ that satisfies the condition $b_k=0$ in the equation
\begin{equation}
    \eigvec =\eigvec_1 \mathrm{cos}(\theta) + \eigvec_2 \mathrm{sin}(\theta) \ .
\end{equation}
The current eigenvector is computed by rotating the reference vector by the new angle $\theta +\Delta \theta$ as given by the search direction. Having computed the new eigenvector in this way, the material interfaces are updated accordingly, and new equilibrium positions are computed using forward simulation.

\section{Results}
\label{sec:results}

We evaluate our method on a range of examples that demonstrate its potential to solve general inverse design problems for elastic surfaces structured with bi-material distributions. Statistics for all experiments
are listed in Tab. \ref{tab:optPerformance}.

\paragraph{Controlling Macromechanical Properties}
In our first example, we explore stripe patterns as a design space for modulating the mechanical properties of an isotropic base material. To this end, we initialize our method with a radial vector field that generates concentric stripes on a square patch. We impose periodic boundary conditions and use homogenization to compute directional stiffness profiles based on the formulation by Schumacher et al. \shortcite{Schumacher18Mechanical}.
As can bee seen from Fig. \ref{fig:FabricModulation}(\textit{a}), despite its apparent symmetry, the tiling of this pattern is not rotationally symmetric and its stiffness is thus not isotropic. Given an isotropic stiffness profile as target, however, our method finds a concentric but slightly more rectangular pattern that leads to the desired behavior (Fig. \ref{fig:FabricModulation}(\textit{b})). As shown in Fig. \ref{fig:FabricModulation}(\textit{c}) and (\textit{d}), our method is likewise able to find modified patterns that yield orthotropic and tetragonal behavior, respectively.

\paragraph{Variable Stiffness Materials}
In a second example, we use our method to design variable stiffness materials. We consider a rectangular patch clamped at two opposite boundaries as shown in Fig. \ref{fig:GradedStiffness}. As design objective, we ask that displacements should vary linearly when applying a constant horizontal force density such as to encourage a stiff to soft distribution from the top to the bottom boundary. We start with a initial material distribution of vertical stripes, for which the stiffness---measured as applied force magnitude divided by resulting displacement norm---is the same for both top and bottom boundaries. 
As can be seen from Fig. \ref{fig:GradedStiffness} (\textit{b, right}), our method produces a pattern that yields the desired stiffness gradient. Visually, stripes have changed direction on the stiffer boundary (\textit{top}) into an almost horizontal arrangement, whereas they remain largely vertical on the softer end (\textit{bottom}). This non-trivial transition is enabled through a sequence of turning and branching points that, despite rather substantial changes, lead to an overall smooth and continuous pattern. 
\begin{figure*}[h]
    \centering
    \includegraphics[width=\textwidth]{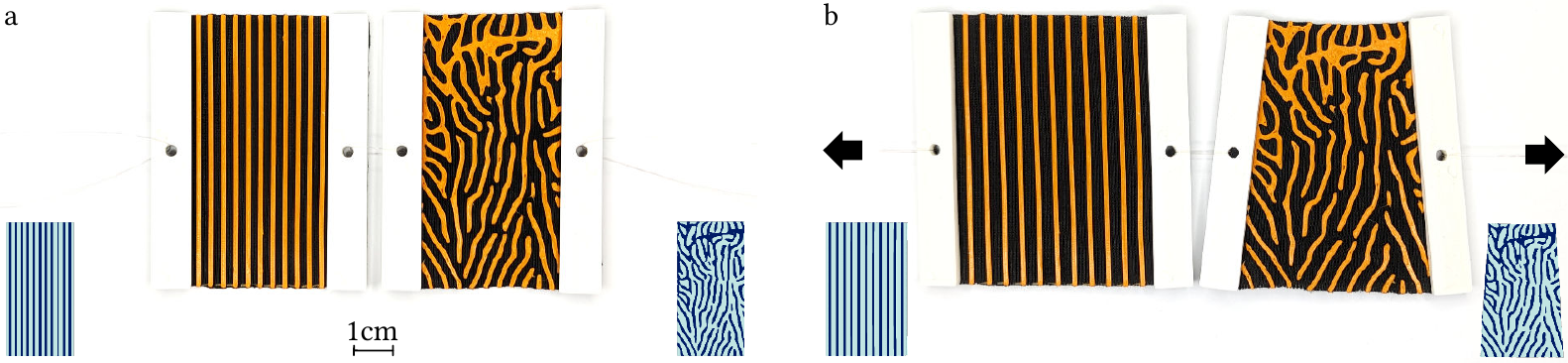}
    \caption{Variable Stiffness Design. Comparison between an initial design with constant stiffness (\textit{a, left}) and our optimized design with linearly varying stiffness (\textit{a, right}). When applying the same force density, our optimized design (\textit{b, right}) exhibits the desired stiffness gradient whereas the initial design shows constant stiffness as expected.}
    \label{fig:GradedStiffness}
\end{figure*}

\paragraph{Compliant Gripper}
While the previous examples have focused on the control of in-plane mechanical properties, Differentiable Stripe Patterns can likewise be used to drive large out-of-plane deformations upon actuation.
To explore this application, we study an example inspired by Kirigami thin-shell grippers \cite{Yang21Grasping}.
Starting from the parallel stripes design shown in Fig. \ref{fig:teaser} (\textit{a}), we aim to optimize the pattern such that, when pulled along its horizontal axis, the wings of the sheet fold such as to minimize the distance between their tips.
Actuating the initial design with parallel stripes does not result in sufficient lateral deflection (Fig. \ref{fig:teaser} (\textit{c})). For our optimized design, however, the same actuation produces the desired large out-of-plane motion that is able to lift a small 3D-printed model; see Fig. \ref{fig:teaser}(\textit{d}) and the accompanying video.

\paragraph{Soft Pneumatic Actuator}
Our method can also be combined with other means of actuation to produce large controlled deformations. 
We demonstrate such a use case by designing a pneumatic actuator---a basic building block for soft robotics applications.
The actuator is made from a rectangular piece of textile structured with stiff stripe patterns. The reinforced design is connected along its long boundary and sealed at the far ends. We simulate inflation by applying constant pressure forces \cite{Montes20Skintight} and enforce periodic boundary conditions on both displacement and stripe patterns. Our design objective is to reach a target shape corresponding to a constant-curvature bend of $180^{\circ}$. 
As shown in Fig. \ref{fig:PneumaticActuator}, our method produces a smooth stripe pattern that closely approximates the target shape in simulation. One can distinguish two main mechanisms that enable this behavior: ribs running along the radial direction such as to allow axial stretch on the top surface while preventing an increase in radius; and a concentration of stiff material on the bottom surface that leads to differential axial stiffness and, consequently, a preferred bending direction. Our physical prototype confirms the feasibility of this design and achieves the targeted change in end-effector orientation, albeit with a somewhat larger lateral contraction.
\begin{figure*}[h]
    \centering
    \includegraphics[width=\textwidth]{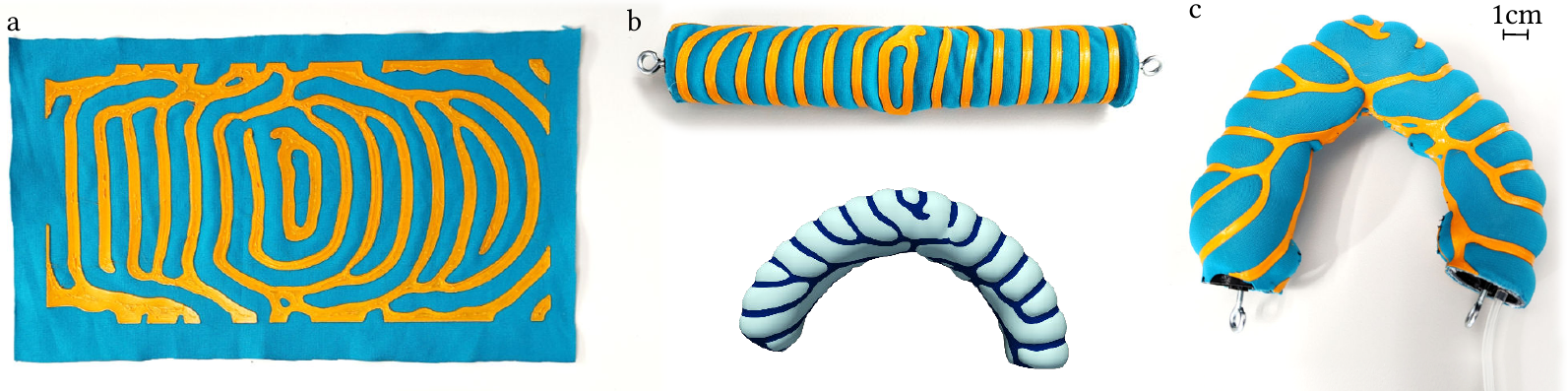}
    \caption{Soft Actuator. Using a stretchable textile as substrate, we optimize for stripe-shaped reinforcements (\textit{a}) such that, once sewn and sealed (\textit{b, top}), the soft actuator deforms into a desired target shape upon pressurization (\textit{c}). The prediction from our simulation model (\Changes{\textit{b, bottom}}) is in good agreement with the physical actuator.}
    \label{fig:PneumaticActuator}
\end{figure*}

\paragraph{Optimization of Insoles}
Stripe patterns are ideal for designing materials with targeted directional deformations, and all examples shown so far relied on this unique capability. To analyze the robustness of our method with respect to design goals outside this space, we investigate an example that does not inherently require material anisotropy.
To this end, we consider the design of a compliant shoe sole that exhibits low stiffness to normal loading in the heel and forefoot regions but offers higher stiffness in the midfoot region for increased stability. As can be seen in Fig. \ref{fig:ShoeInsole}, our method finds a pattern layout that achieves these stiffness goals in simulation by concentrating stiff material in the midfoot region. To evaluate this result quantitatively, the initial design shows average displacements in the normal direction of $-3.66mm$ for heel and forefoot regions and $-6.31mm$ for the midfoot. After optimization, these values have changed to $-4.67mm$ and $-2.52mm$, respectively, indicating that the desired stiffness distribution has been achieved.
It can further be observed that the quality of the stripes deteriorates as the optimization tries to increase the amount of stiff materials. Our smoothness and singularity regularizers put a bound to this trend and largely succeed in maintaining stripe integrity.  

\begin{figure*}[h]
    \centering
    \includegraphics[width=\textwidth]{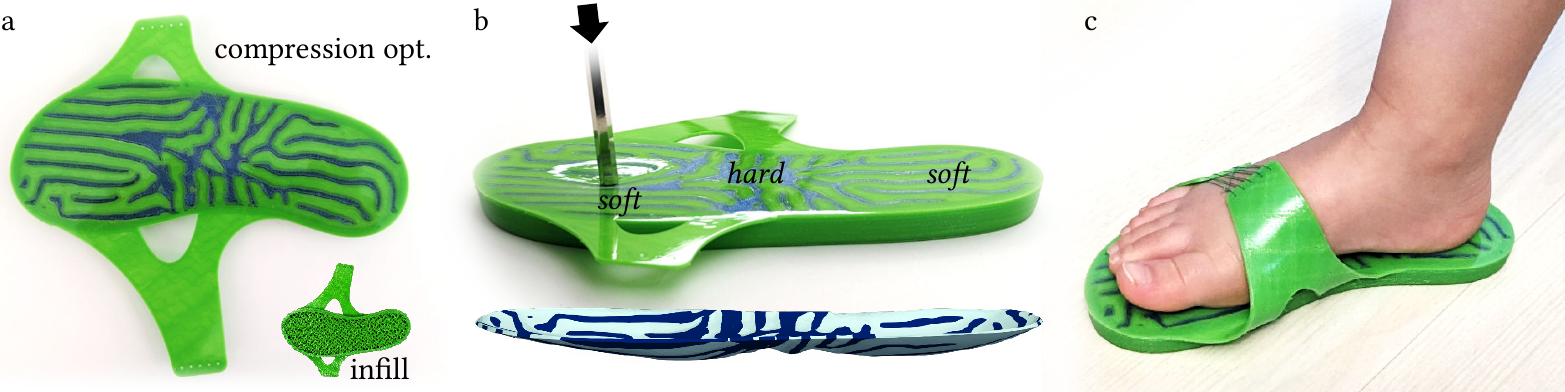}
    \caption{Heterogeneous stiffness optimization. We optimize for stripe patterns such as to produce a shoe insole with a soft response to vertical loading in the heel and forefoot region while offering larger stiffness in the midfoot region.  
    A bi-material stripe pattern with soft TPU (\textit{green}) and stiff PET (\textit{blue}) is placed on top of a soft infill (\textit{a}). The stripe pattern is optimized to yield a soft response to vertical loading in the heel and forefoot region while offering larger stiffness in the midfoot region (\textit{b}) such as to increase comfort when worn (\textit{c}).}
    \label{fig:ShoeInsole}
\end{figure*}

\begin{figure}[h]
    \centering
    \includegraphics[scale=1.0]{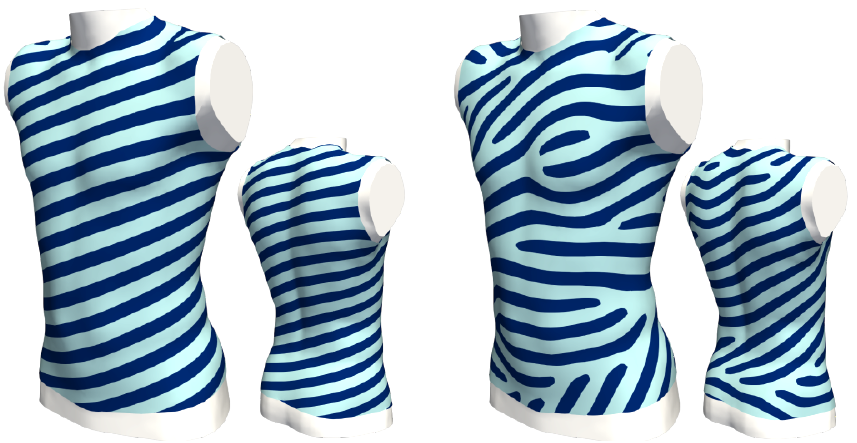}
    \caption{Functional Elastic Shirt. Using a stretchable textile as substrate, we optimize for stripe-shaped reinforcements such as to minimize the strain energy of the garment. \textit{Left}: initial design with parallel stripes in front and back view. Our optimized design (\textit{right}) yields a significant decrease in energy. }
    \label{fig:StripeShirt}
\end{figure}

\paragraph{Garment Design}
Fitted sportswear and medical garments rely on the ability to locally control stretch and stiffness.
To investigate the potential of Differentiable Stripe Patterns in this context, we study the problem of designing an elastic sports shirt structured with stripe-shaped reinforcements. The shirt is fitted onto a torso model and, compared to its rest shape, must increase in area by about 10\% to conform to the body. Stretch is therefore inevitable and we aim to route stripes such as to minimize the elastic energy of the shirt.  
The initial  design with parallel stripes shown in Fig. \ref{fig:StripeShirt}(\textit{left}) exhibits a high strain energy value since stiff reinforcements experience high stress with such a layout. Our optimization method finds a stripe pattern that decreases the energy by more than 25\% with stripes that  branch and meander such as to avoid long tension lines and alignment with principal stress directions. See Fig. \ref{fig:StripeShirt}(\textit{right}).

\paragraph{Structural Optimization of Thin Shells}
While we primarily designed our method to operate on bi-material distributions, Differentiable Stripe Patterns can also be used to generate geometric modifications for single material designs. We investigate this \textit{geometry mode} on a structural optimization problem where we seek to increase the stiffness for a simple vase model subject to vertical loading as shown in Fig. \ref{fig:ThinShells_graph}. Instead of defining a bi-material distribution, we use stripe patterns with sinusoidal cross sections to generate normal displacements. 
Structuring surfaces with so called \textit{beads} and \textit{groves} to increase stiffness is a common strategy---also referred to as topography optimization---in engineering design. 
As shown Fig. \ref{fig:ThinShells_graph} (\textit{3}---\textit{4}), the optimized design found by our method yields a structure with significantly improved axial stiffness that shows no sign of buckling even as we increase load by a factor of two (Fig. \ref{fig:ThinShells_graph}). As a visual interpretation, adding stripe-shaped normal displacements means that quasi-isometric bending modes for the plain design would induce substantial in-plane deformation in the optimized design. Thanks to its fully-differentiable nature, our method is able to exploit this effect, leading to structurally efficient and aesthetically pleasing patterns.

\Changes{
\paragraph{Choice of Initial Guess}
As with any nonlinear, nonconvex optimization problem, there can be many local minima in the solution landscape. As a result, the pattern found depends on the choice of initial guess. We explore the effect of different initial guesses on an example that aims to modulate the directional stiffness of a reinforced fabric into a tetragonal profile. Starting from unidirectional and radial vector fields, we optimize for the same target and compare their results.
As can be seen in Fig. \ref{fig:initialGuess}, the symmetry of the concentric circles allows the pattern to comfortably reach the target with very few changes. In contrast, the unidirectional pattern, which exhibits an extreme initial stiffness profile, is forced to break its symmetry to better fit the target, without achieving the same level of success as the concentric circles.
}

\begin{figure}[h]
    \centering
    \includegraphics[scale=1.0]{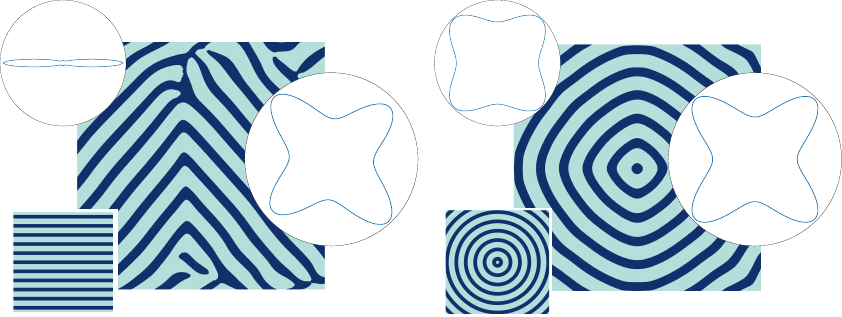}
    \caption{\Changes{
    Influence of initial guess on optimization results.
    We optimize a reinforced fabric to exhibit a tetragonal stiffness profile using a unidirectional (\textit{left}) and a radial (\textit{right}) vector field  as initial guess. Initial patterns and stiffness profiles are shown as bottom and top inset images, respectively. }}
    \label{fig:initialGuess}
\end{figure}

\begin{figure*}[h]
    \centering
    \includegraphics[width=\textwidth]{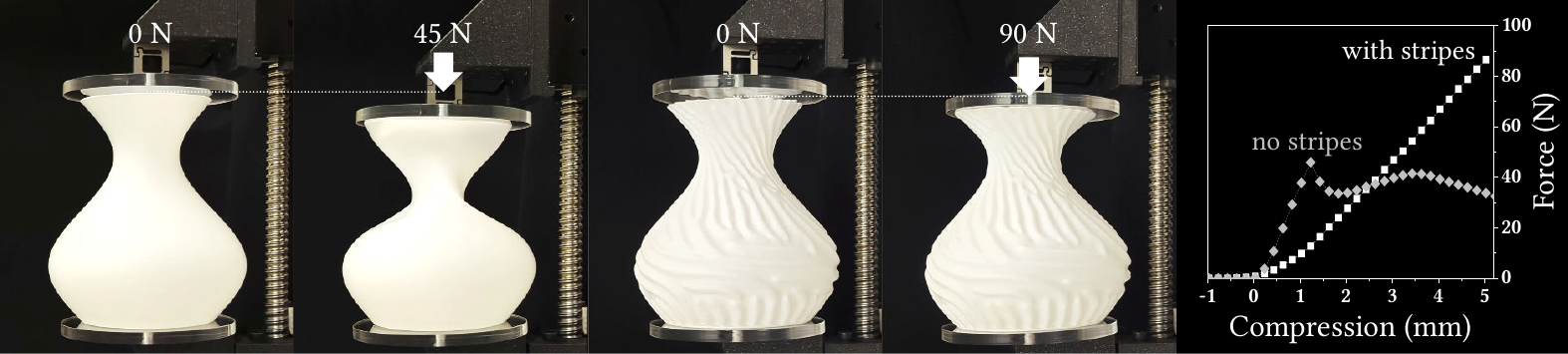}
    \caption{ Structural Optimization with Differentiable Stripe Patterns. From \textit{left} to \textit{right}: a 3D-printed thin-walled vase (\textit{1}) buckles under a vertical load of 45N (\textit{2}). To prevent this failure mode, we optimize for normal offsets in the form of stripe patterns such as to maximize the stiffness of the vase with respect to vertical loading (\textit{3}). Our optimized design deforms much less under vertical loading (\textit{4}) and shows no signs of buckling even for 90N (\textit{5}). 
    }
    \label{fig:ThinShells_graph}
\end{figure*}

\begin{table*}[t]
	\centering
	\caption{Summary of parameters and performance statistics for all our experiments.}
        \begin{tabular}
        {p{2.0cm}p{1.1cm}p{2.0cm}p{2.2cm}p{2.2cm}p{2.0cm}p{2.0cm}}
		\toprule
		Example  & \# dofs & \# iterations & Time per iter. [s] & Initial objective & Final objective & Figure\\
		\midrule

 Isotropic      	&		10481		&			55 				&		21.41				&		391.81			&	38.04 & Fig. \ref{fig:FabricModulation}(b)\\
Orthotropic		&		10481		&			310				&		21.22			&			7324.53			&	60.06 & Fig. \ref{fig:FabricModulation}(c) \\
Tetragonal		&		10481			&		267				&		21.50			&			1721.02		&		21.92 & Fig. \ref{fig:FabricModulation}(d) \\
Cloth			&		8138			&		66				&		193.97			&			776.17		&		642.59 & Fig. \ref{fig:StripeShirt}	\\					
Vase		&			14026			&		592			&			25.68			&			2344.09		&		1957.59 & Fig. \ref{fig:ThinShells_graph} \\
Actuator		&		11689			&		384				&		34.27			&			3574.64	&		901.40 & Fig. \ref{fig:PneumaticActuator} \\

Gripper			&		6718           &           14           &  389.31                 &         	1486.14     &	311.59   & Fig. \ref{fig:teaser}          \\
Var. Stiffness        &    11592 &  139 &	130.13	&	1031.89  &	12.49    & Fig. \ref{fig:GradedStiffness}   \\
Insole			&		6681			&		357			&			126.10			&			209.17	&		103.76  & Fig. \ref{fig:ShoeInsole}\\
		\bottomrule
	\end{tabular}
	\label{tab:optPerformance}
\end{table*}


\section{Conclusion}

We presented Differentiable Stripe Patterns, a computational approach that unlocks Stripe Patterns as a design space for structured surfaces.
In order to invert the stripe pattern pipeline  due to Kn{\"o}ppel et al.  \shortcite{Knoeppel15Stripe}, we have addressed several key challenges. 
First, we resolved ambiguities due to eigenvalue multiplicity by establishing a unique parameterization of the corresponding eigenplane, resulting in well-defined eigenvector derivatives. Second, to allow for accurate modeling of moving material interfaces, we proposed a combination of solid shells and extended finite elements. Finally, we introduced design space regularizers to avoid numerical singularities and improve stripe neatness. We combined these components with equilibrium state derivatives into an end-to-end differentiable pipeline that enables inverse design with high-level performance objectives. Our results indicate that stripe patterns are indeed a promising design space for bi-material surfaces, and that gradient-based optimization is an effective tool for exploring this space. There are, nevertheless, several limitations and corresponding opportunities for future improvements.

\subsection{Limitations}

We have assumed that our bi-material designs have material interfaces extending through the entire thickness. This assumption is a somewhat coarse approximation of our reinforced fabric examples, for which we added 3D-printed stripes \textit{on top} of a textile substrate.
While our results showed fairly good accuracy compared to real-world experiments, Jourdan et al. \shortcite{Jourdan2022Simulation} have recently demonstrated that a thin shell model can be adapted to reflect this non-symmetric patterning. We expect that similar modifications could likewise be effective for our solid shell model. Another alternative would be to use a bi-layer solid shell, with a homogeneous bottom layer and a structured top layer consisting of stiff and void material.

Our method is currently limited to a single linear cut per element. This is not much of a restriction in practice, since our focus is on low-frequency stripe patterns and we can simply choose mesh resolution accordingly. An exception are branches, which involve two cuts in a single element that form an inward dent. Our method currently ignores this non-convex feature and approximates it with its convex closure. Nevertheless, extending our method to handle multiple cuts per element using the hierarchical scheme by Zehnder et al. \shortcite{zehnder2017metasilicone} should pose no major problems.

We believe that the combination of solid shells and XFEM offers  unique advantages in terms of accuracy for computation time. In this work, we only explored constant thickness shells with two materials. However, solid shells could be useful for many other graphics applications where through-the-thickness deformation is an important effect, including simulation of skin, faces, and filled garments.

Finally, we have only started to explore the possibilities of stripe patterns for functional garment design. In the future, we would like to investigate the design of haptic garments that use stripe patterns to provide feedback on, e.g., posture \cite{Vechev2022Computational}. Extending passive stripes toward active control and sensing is likewise an exciting avenue for future work.

\begin{acks}
We are grateful to the anonymous reviewers for their valuable comments. This work was supported by the European Research Council (ERC) under the European Union’s Horizon 2020 research and innovation program (grant agreement No. 866480), and the Swiss National Science Foundation through SNF project grant 200021\_200644.
\end{acks}

\bibliographystyle{ACM-Reference-Format}
\bibliography{references}

\end{document}